\documentclass[a4paper,11pt]{article}
\usepackage{amssymb,amsmath,latexsym,amsthm}
\usepackage{color,mathrsfs}
\usepackage{url}
\usepackage{enumerate}
\usepackage[square,sort,comma,numbers]{natbib}
\usepackage{dsfont}
\usepackage{diagbox}
\usepackage{graphicx}
\usepackage{colortbl}

\usepackage{endnotes}

\renewcommand{\textbf}[1]{\begingroup\bfseries\mathversion{bold}#1\endgroup}

\newlength{\bibitemsep}\newlength{\bibparskip}
\let\oldthebibliography\thebibliography \renewcommand\thebibliography[1]{
  \oldthebibliography{#1} \setlength{\parskip}{\bibitemsep}
  \setlength{\itemsep}{\bibparskip} }

\usepackage{fancyhdr}
\usepackage{pstricks,pst-plot,pst-node,pstricks-add}

\usepackage[left=1.95cm, right=1.95cm, top=2cm]{geometry}

\newtheorem{thm}{Theorem}[section]
\newtheorem{defi}{Definition}[section]
\newtheorem{corollary}[thm]{Corollary}
\newtheorem{prop}[thm]{Proposition}

\theoremstyle{definition}
\newtheorem{remark}[thm]{Remark}

\newcommand{\R}{\mathbb R}

\newcommand{\Z}{\mathbb Z}
\newcommand{\N}{\mathbb N}

\numberwithin{equation}{section}

\def\XXint#1#2#3{{\setbox0=\hbox{$#1{#2#3}{\int}$}
    \vcenter{\hbox{$#2#3$}}\kern-.5\wd0}}

\allowdisplaybreaks
\date{date}
\begin{document}
\title{On energy ground states among crystal lattice structures with prescribed bonds}
\author{Laurent B\'{e}termin\\ \\
Faculty of Mathematics, University of Vienna,\\ Oskar-Morgenstern-Platz 1, 1090 Vienna, Austria\\ \texttt{laurent.betermin@univie.ac.at}. ORCID id: 0000-0003-4070-3344 }
\date\today
\maketitle

\begin{abstract}
We consider pairwise interaction energies and we investigate their minimizers among lattices with prescribed minimal vectors (length and coordination number), i.e. the one corresponding to the crystal's bonds. In particular, we show the universal minimality -- i.e. the optimality for all completely monotone interaction potentials -- of strongly eutactic lattices among these structures. This gives new optimality results for the square, triangular, simple cubic (SC), face-centred-cubic (FCC) and body-centred-cubic (BCC) lattices in dimensions 2 and 3 when points are interacting through completely monotone potentials. We also show the universal maximality of the triangular and FCC lattices among all lattices with prescribed bonds. Furthermore, we apply our results to Lennard-Jones type potentials, showing the minimality of any universal minimizer (resp. maximizer) for small (resp. large) bond lengths, where the ranges of optimality are easily computable. Finally, a numerical investigation is presented where a phase transition of type ``square-rhombic-triangular" (resp. ``SC-rhombic-BCC-rhombic-FCC") in dimension $d=2$ (resp. $d=3$)  among lattices with more than 4 (resp. 6) bonds is observed.
\end{abstract}


\section{Introduction}

The origin of periodic order in crystalline solid still remains a mystery, especially when we look from the mathematical side of the problem. The so-called ``Crystallization Conjecture" (see \cite{BlancLewin-2015}) states that most of the ground states of Physical interacting systems, especially at low temperature, are periodic lattices. Solving this type of problem for a given interaction energy and for systems with a large number of interacting points is mathematically and computationally extremely challenging in dimension $d\geq 2$ due to the very high number of variables and critical points. Even in the simple Born-Oppenheimer adiabatic approximation case where the interactions between atoms are pairwise, only few results showing the minimality of the triangular lattice \cite{Rad2,Crystal,Luca:2016aa} and the square lattice \cite{BDLPSquare} for the total interaction energy are available. Furthermore, adding angular-dependent interactions allows to obtain the square lattice \cite{Stef1}, the honeycomb lattice \cite{ELi,Stef2} and the Face-Centred-Cubic (FCC) lattice \cite{TheilFlatley} as ground states (see Table \ref{fig:lattices} for a precise definition of these lattices).

\medskip

When restricted to simple periodic configurations, the number of degrees of freedom of the system can be reduced, leading to more general optimality results that give important clues concerning the possible ground states for the corresponding general crystallization problem. Rigorous optimization of crystal's shapes under constraints has been performed for inverse power laws \cite{Rankin,Eno2,Cassels,Diananda,Ennola,FieldsEpstein,Coulangeon:kx}, Gaussians \cite{Mont,Coulangeon:2010uq,BeterminPetrache,LBSoftTheta20,BFMaxTheta20,CKMRV2Theta}, Lennard-Jones potentials \cite{BetTheta15,Beterloc,Beterminlocal3d,SamajTravenecLJ} and more general functions \cite{OptinonCM,LBMorse} or energies \cite{CheOshita,Sandier_Serfaty,LuoChenWei,LuoWeiBEC,LBEAM21}. Most of them investigate the minimization problem at fixed density and look for optimality results for the usual best (densest) lattice packing that are the triangular lattice, the FCC lattice, the Gosset lattice $\mathsf{E}_8$ and the Leech lattice $\Lambda_{24}$. In dimension $d=2$, the triangular lattice has been shown to be universally optimal, i.e. minimal for any completely monotone interaction potential (see Definition \ref{def:univopt}), among lattices of fixed density by Montgomery in \cite{Mont} and is conjectured to be universally optimal among all periodic configurations by Cohn and Kumar \cite{CohnKumar}. In dimension $d=3$, there is no universal minimizer at fixed density but the FCC and Body-Centred-Cubic (BCC) lattices are believed to be the only possible minimizers for Gaussian interactions according to Sarnak-Strombergsson Conjecture (see \cite[Eq. (43)]{SarStromb}). Moreover, $\mathsf{E}_8$ and $\Lambda_{24}$ have been recently shown in \cite{CKMRV2Theta} to be universally optimal among all periodic configurations of fixed density in dimensions $d\in \{8,24\}$.

\medskip

In this paper, our goal is to derive universal optimality results among lattices with prescribed minimal vectors, i.e. when their number and side length are fixed. The segments between lattice's nearest-neighbors (i.e. achieving the lattice's smallest distance) will be called ``bonds" and the number of nearest-neighbors of any atom will be called ``coordination number" as in the Solid State Physics setting (see Definition \ref{def:bonds}). Physically, the ion-electron interaction is responsible for fixing the length of bonds (where the interaction energy is minimal) and a specific coordination number is usually the consequence of more particles interactions (not only pair). Therefore, the model we are investigating in this work is not only derived from the Born-Oppenheimer approximation but can be viewed as a mixed one: taking into account the ion-electron interaction as well as other forces inducing a specific coordination number, we only consider a first order approximation of the total energy (i.e. a pairwise one) that we want to optimize among lattices.

\medskip

Furthermore, it has to be noticed that the common point of the crystallization's proofs cited above is the fact that interaction potentials induce one or several  ``preferred" equilibrium distances (and/or angles) between particles in the ground state configurations. This yields to a locally stable configuration which, combined with packing arguments and long-range interactions estimates, appears to be globally minimal. In other words, nearest-neighbors interactions play a fundamental role in crystal's formation as well as in the determination of its optimal shape, leading to rigidity properties (see also \cite{LBAngleRigidity} for a related problem with fixed bonds and angles). Physically, it is a natural phenomenon (see e.g. \cite[Sect. 1.2]{Kaxiras}) due to a combination of repulsion between particles at short distance (Pauli Principle) and attraction at long distance which can be due to the nature of the particles (e.g. van der Waals forces), the environment (e.g. a background of opposite charges in a Jellium, see \cite{ParrYang,GiulianiVignale}) or some external confinement. Let us quote Feynman about the importance of local rules in the crystal's formation at low temperature (see \cite[Chap. 30]{FeynmanII}): \textit{``when the atoms of matter are not moving around very much, they get stuck together and arrange themselves in a configuration with as low an energy as possible. If the atoms in a certain place have found a pattern which seems to be of low energy, then the atoms somewhere else will probably make the same arrangement. For these reasons, we have in a solid material a repetitive pattern of atoms."} (see also \cite{Dolbilin} for details about local rules). The resulting bonds created between particles of the same species appear to have the same length. In this work, we will therefore assume at the same time simple periodicity and the fact that each particle is exactly bonded to $r$ other particles with bond length $\lambda$, and we actually seek to minimize potential energies of type
\begin{equation}\label{eq:Efintro}
E_f[L]:=\sum_{p\in L\backslash \{0\}} f(|p|^2)
\end{equation}
among such periodic point configurations $L$ where particles are interacting through a radially symmetric potential $f$ decaying sufficiently fast to zero at infinity. The underlying physical question is therefore the following: 
\begin{itemize}
\item \textbf{Physical Problem.} Assuming that strong forces (like in metals) obliges the bonds to have a certain fixed length as well as an exact (or minimum) coordination number for each atom, what is the crystal lattice structure with the lowest potential energy?
\end{itemize}
 Notice that, for lattices -- also called perfect lattices (see e.g. \cite{Gruber}) -- that are the unique one to have a certain coordination number $r$, like for instance the triangular, FCC, $\mathsf{E}_8$ or Leech lattices in dimensions $d\in \{2,3,8,24\}$, this minimization problem is obviously not relevant. However, comparing their energies with the one of other lattices having the same minimal bond length $\lambda$, but a different coordination number $r$, is of interest. 

\medskip

More precisely, our work is aimed to generalize the one of Fields \cite{FieldsEpstein} who focused on universal local minimizers for inverse power-law interactions, i.e. where $f(r)=r^{-s/2}$ and $E_f$ given by \eqref{eq:Efintro} is the Epstein zeta function (see \eqref{eq:zetatheta}). Fields was in particular interested in ``fragile" lattices (see \cite{FieldsFragiles}) reaching a local minimum of the sphere packing density among lattices with a fixed coordination number. The most symmetric and important fragile lattices are, in dimensions $d\in \{2,3\}$, the square, Simple Cubic (SC) and BCC lattices, for which only few minimality results are actually available. Furthermore, together with the triangular and FCC lattices (those reache the maximum of the sphere packing density), they are known as the only strongly eutactic lattices (see Definition \ref{def:eutactfragile} and \cite[Def. 1]{Gruber}) in dimensions $2$ and $3$. For instance, excepted the very recent crystallization result obtained in \cite{BDLPSquare} as well as the preliminary works done in \cite[Sect. 6.2]{OptinonCM} for designing a particular potential having a square lattice as a ground state, only local minimality properties for the square lattice have been derived (see e.g. \cite{Beterloc}). The situation for the BCC lattice is similar, see e.g. \cite{Ennola,Beterminlocal3d}. Concerning the SC lattice, only numerical investigations are available (see e.g. \cite{MarcotteZ3}), showing that a pairwise potential can be designed to have a SC lattice as a ground state. Moreover, in the multicomponent case, the rock-salt structure consisting of a SC lattice structure with an alternation of charges $\pm 1$ at its sites is expected to be a ground state among lattice structures and charges for a large class of potentials (see \cite{LBMFHK20}), as already shown in dimension $d=2$ by Friedrich and Kreutz \cite{FrieKreutSquare} in the case of short-range potentials with only two types of charges $\pm 1$.

\medskip

Two kinds of interaction potentials $f$ are used in this paper: the completely monotone functions that are Laplace transform of nonnegative measure (see Definition \ref{def:potenergy}) and the Lennard-Jones type potentials defined by
\begin{equation}\label{eq:LJintro}
f(r):=\frac{a}{r^p}-\frac{b}{r^q},\quad (a,b)\in (0,\infty),\quad p>q>\frac{d}{2}.
\end{equation}
This type of potential was actually initially proposed by Mie \cite{Mie}, then popularized by Lennard-Jones \cite{LJ} and now widely used for molecular simulations (see e.g. \cite{MolecSimul}). They are one-well potentials and the most classical parameter's choice is the one taking into account the van der Waals attraction at large distance, i.e. $f(r^2)=r^{-12}-2r^{-6}$ with minimum at $r_0=1$. Whereas the crystallization problem associated to  \eqref{eq:LJintro} is still open -- the asymptotic minimizer as the number of particles goes to infinity is conjectured to be a triangular lattice, see e.g. \cite[Fig. 1]{BlancLewin-2015} --, minimizing $E_f$ among lattices has been recently intensively studied. General minimality results have been derived, showing the optimality of universal minimizers at fixed high density and global optimality of a triangular lattice for small exponents $(p,q)$, see e.g. \cite{BetTheta15,BDefects20}. Furthermore, local optimality results at fixed density are showed in \cite{Beterloc,Beterminlocal3d} in dimensions 2 and 3 for square, triangular, SC, FCC and BCC lattices, as well as numerical investigation of the minimizers at fixed density (see also \cite{SamajTravenecLJ} and our Tables \ref{fig:locLJ2d} and \ref{fig:locLJ3d}). In particular, the two-dimensional minimization problem at fixed density exhibits a phase transition of type ``triangular-rhombic-square-rectangular" for its minimizer as the density decreases. We aim to investigate the same kind of minimization problem when the space of lattices with fixed density is replaced by the one with fixed bonds.
 
\medskip

The minimization problem we are investigating in this work appears indeed to be different than the one among periodic structures with fixed density. In dimension $d\in \{2,3\}$, whereas the latter mostly favors the perfect triangular and FCC lattices as minimizers, the one we are studying here is dedicated to new optimality results for the fragile square, SC and BCC lattices (see again \cite{FieldsEpstein} for the inverse power-law case). While exploring the space of lattices with prescribed minimal vectors, we easily see that the lattice's density of points changes continuously. Therefore, studying this problem is complementary to the fixed density case in order to understand the variations of $E_f$ in the space of all $d$-dimensional lattices. Also, from a technical point of view, the minimization of lattice energies is done here with linear constraints on the side length whereas the constraint is polynomial in the fixed density case (the discriminant of the quadratic forms is fixed). This is a very important detail that allows us to use convexity arguments on a convex set, getting the uniqueness of our (a priori local) minimizers.

\medskip

Hence, using techniques developed in \cite{DeloneRysh,FieldsEpstein,Gruber,BetTheta15}, we show the following results:
\begin{itemize}
\item \textbf{Universal minimizers.} Any strongly eutactic lattice (see Definition \ref{def:eutactfragile}) is a unique universal minimizer among lattices with corresponding prescribed bonds. In dimension $d\in \{2,3\}$, the square, triangular, SC, FCC and BCC lattices are the only minimizers of this type (see Theorem \ref{thm:univlocal});
\item \textbf{Universal maximizers in dimensions 2 and 3.} The triangular and FCC lattices are universal maximizers among lattices with prescribed bonds (see Corollary \ref{cor:2d3d});
\item \textbf{Optimality at small and large scales for Lennard-Jones type potentials.} We now assume that $f$ is a Lennard-Jones type potential defined by \eqref{eq:LJintro}. Any universal minimizer is again minimal for $E_f$ among lattices with prescribed bonds of side length $\lambda\in (0,\lambda_0)$, where $\lambda_0$ can be explicitly computed. Furthermore, any unique universal maximizer is minimal for $E_f$ among lattices with prescribed bonds of side length $\lambda>\lambda_1$, where $\lambda_1$ can be again explicitly computed (see Proposition \ref{prop:LJ}).
\end{itemize}

Additionally, we have numerically studied the two-dimensional classical $(12,6)$ Lennard-Jones energy with potential $f(r^2)=r^{-12}-2r^{-6}$ among all lattices with bonds side length $\lambda>0$ (see Section \ref{subsec:num}). Its lattice minimizer exhibits a phase transition of type ``square-rhombic-triangular" as $\lambda$ increases (``rhombic" designing a non-triangular non-square lattice with $r\geq 4$ prescribed bonds). In Table \ref{fig:locLJ2d}, we give the numerical values we found and we compare with the fixed area minimization problem (the area of a lattice being the inverse of its density, i.e. the volume of its unit cell) for the very same Lennard-Jones potential already studied in \cite{Beterloc}. We observe that, for small $\lambda$, the square lattice replaces the triangular one as a minimizer. Then they are both triangular in the range of $\lambda$ where the global minimizer of $E_f$ among all possible lattices is achieved and actually expected to be the triangular lattice $\lambda_{opt} \mathsf{A}_2$ where $\lambda_{opt}\approx 1.112$. Then the minimizer stays triangular whereas, in the density fixed case, it started to change from rhombic to degenerate rectangular lattice.

\medskip

\begin{table}[!h]
\centering
\begin{tabular}{|c|c|c|}
\hline
 & \textbf{Min. at fixed area $A$} & \textbf{Min. at fixed minimal bond $\lambda$} \\
\hline
\textbf{Square phase} & $A_1\approx 1.143<A<A_2\approx 1.268$ & $0<\lambda  < \lambda_0\approx 0.76286$\\
\hline
\textbf{Triangular phase} & $0<A<A_{BZ}\approx 1.138$ & $\lambda> \lambda_1\approx 0.989$ \\
\hline
\textbf{Rhombic phase} & $A_{BZ}<A<A_1$ & $\lambda_0<\lambda<\lambda_1$ \\
\hline
\textbf{Rectangular phase} & $A> A_2$ & None\\
\hline
\end{tabular}
\caption{Comparison of minimization problems among lattices with fixed area $A$ (see \cite{Beterloc}) and prescribed bonds with side length $\lambda$ we have found for $E_f$ (see Section \ref{subsec:num}) where $f(r)=r^{-6}-2 r^{-3}$, i.e. $f(r^2)$ is the classical Lennard-Jones potential. Remark that a square lattice of area $A_1$ or $A_2$ has a side length equal to $\sqrt{A_1}\approx 1.069$ or $\sqrt{A_2}\approx 1.126$, and a triangular lattice of area $A_{BZ}$ has a side length equal to $\ell_{BZ}\approx 1.146$.}
\label{fig:locLJ2d}
\end{table}

Finally, in the three-dimensional case also studied in Section \ref{subsec:num}, the Lennard-Jones energy minimizer exhibits a phase transition as $\lambda$ increases (see Table \ref{fig:locLJ3d}) of type ``SC-rhombic-BCC-rhombic-FCC" (``rhombic" designing here a non-SC non-BCC non-FCC lattice with $r\geq 6$ prescribed bonds). There is currently no complete study of the global minimizer of $E_f$ at fixed density in dimension $d=3$ available, but only local minimality results are showed in \cite{Beterminlocal3d}. If we compare our findings with the latter, the FCC and BCC phases identified in \cite{Beterminlocal3d} are again replaced by a SC one at high density (i.e. small $\lambda$). The global minimizer of our energy among all side lengths $\lambda>0$ and all lattices with bonds of length $\lambda$ is a FCC lattice $\lambda_{opt} \mathsf{D}_3$ where $\lambda_{opt}\approx 0.97$.

\medskip

\begin{table}[!h]
\centering
\begin{tabular}{|c|c|c|}
\hline
 & \textbf{Loc. min. at fixed volume $V$} & \textbf{Min. at fixed minimal bond $\lambda$} \\
\hline
\textbf{SC phase} & $V_1\approx 1.2<V<V_2\approx 1.344$ & $0<\lambda  < \lambda_0\approx 0.735$\\
\hline
\textbf{BCC phase} & $0<V<V_{\min}\approx 1.091$ & $\lambda_1\approx 0.9<\lambda<\lambda_2\approx 0.94$ \\
\hline
\textbf{FCC phase} & $0<V<V_{\min}\approx 1.091$ & $\lambda >\lambda_3\approx 0.95$ \\
\hline
\textbf{Rhombic phase} & Not studied & $\lambda \in (\lambda_0,\lambda_1)\cup (\lambda_2,\lambda_3)$\\
\hline
\end{tabular}
\caption{Comparison of minimization problems among lattices with fixed volume $V$ (only local minimality results, see \cite{Beterminlocal3d}) and prescribed bonds with side length $\lambda$ we have found (see Section \ref{subsec:num}) for $E_f$ where $f(r)=r^{-6}-2 r^{-3}$, i.e. $f(r^2)$ is the classical Lennard-Jones potential. Notice that SC lattices with volume $V_1$ or $V_2$ have side lengths equal to $V_1^{1/3}\approx 1.06266$ or $V_2^{1/3}\approx 1.10357$. Furthermore, FCC or BCC lattices of volume $V_{\min}$ have side lengths $\lambda_{\mathsf{D}_3}\approx 1.15553$ or $\lambda_{\mathsf{D}_3^*}\approx 1.12326$.}
\label{fig:locLJ3d}
\end{table}

Even though we chose to focus on summable interaction potentials, it is expected that our results can be generalized to any completely monotone potential $f$ as well as any Lennard-Jones type potentials with arbitrary small exponents. We have for example in mind the Coulomb potentials $f(r^2)\in \{-\log r, r^{-d+2}, d\geq 3\}$. The main issue is the way to define a new energy $E_f$ (compared to \eqref{eq:Efintro}) when $f$ is non-summable. A classical possibility is to use Ewald Summation Method, writing $E_f$ in terms of lattice theta functions (defined by \eqref{eq:zetatheta}) as it was done for instance in \cite{SaffLongRange}. Since the SC, FCC and BCC lattices are universal minimizers, i.e. minimizers for all lattice theta functions (see Definition \ref{def:univopt}), there is no doubt that the same optimality results will hold for such potentials (see also \cite{PetracheSerfatyCryst} for another connection between universal minimality and non-summable potentials).

\medskip

To conclude, our results are complementary to the one found in \cite{BetTheta15,Beterloc,Beterminlocal3d,SamajTravenecLJ} for the fixed density case and give a new point of view on the lattice energy minimization problem. It also provides new clues on the reason of the SC, FCC and BCC ground states observed in nature for classical crystals (compared to quantum crystals \cite{BuchananQC} where quantum effects cannot be ignored) since these lattices appear to be globally minimal for pairwise interaction potentials among structures with prescribed bonds. For instance, an adequate choice of parameters $(p,q,a,b)$ for a Lennard-Jones type potential can give a model where a real crystal structure is a minimizer of $E_f$ among lattices with a given (the physically true one) minimal bond length $\lambda$. In particular, a SC structure (like the one of $\alpha$-Polonium) can be stabilized (and not only locally) for not too small bonds and for a classical potential which is a novelty compared to previous works (see e.g. \cite{MarcotteZ3,Beterminlocal3d}). However, it is not clear whether a single Lennard-Jones type potential can explain different ground state structures simultaneously. Indeed, a quick look to the periodic table of elements, and especially to the covalent radii of the atoms, does not show that BCC bonds are always smaller than FCC bonds (even in one column) as one might expect from our numerical findings in Table \ref{fig:locLJ3d}. To quote Kaplan \cite[p. 3]{Kaplan}: \textit{``It should be mentioned that semiempirical potentials cannot describe intermolecular potential adequately for a wide range of separations. A given potential with parameters calibrated for one property, often describes other properties inadequately, since different physical properties may be sensitive to different parts of the potential curve."} The inadequacy of central potentials for explaining all aspects of a crystal, including its ground state, has been also discussed in \cite{LBMorse} for Morse type potentials.

\medskip

\textbf{Plan of the paper.} We introduce the precise definitions of lattices, potential and energies in the next Section \ref{sec:def}. Universal optimizers are characterized in Section \ref{subsec:localstable}. Furthermore, the theoretical and numerical study of Lennard-Jones type energies is done in Section \ref{subsec:LJ}.

\section{Setting and notations}\label{sec:def}

\subsection{Lattices}

\begin{defi}[Lattices and their quadratic forms]
For any $d\geq 1$, we call $\mathcal{L}_d$ the space of $d$-dimensional lattices, i.e.
$$
\mathcal{L}_d:=\left\{ L\subset \R^d : L=\bigoplus_{i=1}^d \Z u_i, \{u_i\} \textnormal{ basis of $\R^d$}\right\}.
$$
Each lattice $L\in \mathcal{L}_d$ is associated to a quadratic form $Q_L$ defined on $\Z^d$, and therefore a $d\times d$ matrix $A_L$, such that
$$
Q_L(m):=m^t A_L m:=\left| \sum_{i=1}^d m_i u_i \right|^2,\quad \forall m=(m_1,...,m_d)\in \Z^d,\quad \quad L=\bigoplus_{i=1}^d \Z u_i.
$$
\end{defi}

\begin{defi}[Minimal vectors, bonds]\label{def:bonds}
Let $i\geq 1$ and $L\in \mathcal{L}_d$. The $i$-th smallest distance of $L$ is defined by
$$
\lambda_i(L):=\min \{ |p| : p\in L, |p|>\lambda_{i-1}(L)\},
$$ 
where $\lambda_0(L)=0$ by convention. If $|p-q|=\lambda_1(L)$ where $p,q\in L$, we say that $p$ and $q$ are nearest-neighbors, $p-q$ is a minimal vector of $L$ and the segment $[p,q]=\{(1-t)p + t q : t\in [0,1]\}$ is a called a bond of $L$.
\end{defi}

\begin{defi}[Lattices with prescribed minimal vectors]
 For any $\lambda>0$ and $\mathcal{M}\subset (\Z^d)^r$ where $r\in \N$, we say that $L\in\mathcal{L}_d(\mathcal{M},\lambda)$ if
$$
\forall m\in \mathcal{M}, Q_L(m)=\lambda_1(L)^2=\lambda^2 \quad \textnormal{and}\quad  \forall m\in \Z^d\backslash \mathcal{M}\cup \{0\}, Q_L(m)>\lambda^2,
$$
 i.e. all the lattices $L=\bigoplus_{i=1}^d \Z u_i$ of $\mathcal{L}_d(\mathcal{M},\lambda)$ have exactly $r$ minimal vectors $p=\sum_{i=1} m_i u_i$, where $m=(m_1,...,m_d)\in \mathcal{M}$, of length $\lambda$. Given a lattice $L_0\in \mathcal{L}_d$, $\mathcal{M}_{L_0}$ is defined such that $L_0\in \mathcal{L}_d(\mathcal{M}_{L_0},\lambda)$. Moreover, $\overline{\mathcal{L}_d(\mathcal{M},\lambda)}$ denotes the closure of $\mathcal{L}_d(\mathcal{M},\lambda)$ in $\mathcal{L}_d$. 
 
 \medskip
 
 \noindent Furthermore, in dimension $d\in \{2,3\}$, we define the following set of lattices:
 \begin{align*}
&\mathcal{L}_2(\lambda):=\mathcal{L}_2(\{\pm e_i\}_i,\lambda)\cup \mathcal{L}_2(\{\{\pm e_i\}_i, \pm(e_2-e_1)\},\lambda)=\overline{\mathcal{L}_2(\mathcal{M}_{\Z^2},\lambda)},\\
&\mathcal{L}_3(\lambda):=\mathcal{L}_3(\{\pm e_i\}_i,\lambda)\cup \mathcal{L}_3(\{\{\pm e_i\}_i, \pm (e_1+e_2+e_3) \},\lambda)\cup \mathcal{L}_3(\{\{\pm e_i\}_i, \{\pm (e_j-e_i): j>i\}\},\lambda),
\end{align*}
where $\{e_i\}_{1\leq i \leq d}$ is the canonical basis of $\R^d$, i.e. the lattices with 4 or 6 (resp. 6, 8 or 12) minimal vectors of length $\lambda$. 
\end{defi}

\begin{remark}[Coordination and kissing numbers]\label{rmk:kiss}
The integer $r$ is called the coordination number (or the contact number for the associated lattice packing) of $L$. The maximum of such $r$ among all lattices is called the ``kissing number" (see e.g. \cite{ConSloanPacking,Kissing}), which is $r=6$ (resp. $r=12$) in dimension $d=2$ (resp. $d=3$).
\end{remark}

\begin{remark}[Scaling and closure]
It is useful to notice that $\mathcal{M}_{\lambda L}=\mathcal{M}_L$ for all $\lambda>0$ and all $L\in \mathcal{L}_d$. Furthermore, the ``closure" $\overline{\mathcal{L}_d(\mathcal{M},\lambda)}$ of $\mathcal{L}_d(\mathcal{M},\lambda)$ in $\mathcal{L}_d$ has to be understood as an Euclidean set with the Euclidean distance through the usual parametrization of lattices in $\R^{\frac{d(d+1)}{2}}$ (see e.g. \cite{Terras_1988}).
\end{remark}

As we have written the underlying physical problem in the introduction, our corresponding mathematical problem can therefore be stated as follows:

\begin{itemize}
\item \textbf{Mathematical Problem.} For any fixed $\lambda>0$ and $\mathcal{M}\subset (\Z^d)^r$ for some $r>0$, what is the minimizer of $E_f$ in $\mathcal{L}_d(\mathcal{M},\lambda)$ as well as in $\mathcal{L}_d(\lambda)$?
\end{itemize}

In Table \ref{fig:lattices}, we recall the definition of some important lattices of unit side length as well as their minimal vectors expressed in the canonical basis $\{e_i\}_{1\leq i \leq d}$ of $\R^d$. We will call ``rhombic" any lattice which is not in this list and has more than 4 (resp. 6) bonds in dimension $d=2$ (resp. $d=3$).

\begin{table}[!h]
\centering
\begin{tabular}{|c|c|c|}
\hline
\textbf{Name} & \textbf{Notation and Definition} & \textbf{Minimal vectors $\mathcal{M}\subset (\Z^d)^r$} \\
\hline
SC & $\Z^d= \bigoplus_{i=1}^d \Z e_i$  & $\{ \pm e_i\}_{1\leq i \leq d}$, $r=2d$ \\
\hline
Triangular  &  $\mathsf{A}_2:= \Z\left(1,0 \right)\oplus \Z\left( \frac{1}{2},\frac{\sqrt{3}}{2} \right)  $ & $\{\{\pm e_i\}_i, \pm (e_2-e_1)\}$, $r=6$ \\
\hline
FCC & $\mathsf{D}_3:=\frac{1}{\sqrt{2}}\left[\Z(1,0,1)\oplus \Z(0,1,1)\oplus \Z(1,1,0)  \right]$  &$\{\{\pm e_i\}_i, \{\pm (e_j-e_i)\}_{j>i}\}$, $r=12$ \\
\hline
BCC & $\mathsf{D}_3^*:=\frac{1}{\sqrt{3}}\left[\Z(1,1,-1)\oplus \Z(-1,1,1)\oplus \Z\left(1,-1,1  \right)  \right]$ & $\{\{\pm e_i\}_i, \pm (e_1+e_2+e_3) \}$, $r=8$\\
\hline
\end{tabular}
\caption{Description of the four important types of lattices belonging to $\mathcal{L}_d(\mathcal{M},1)$, $d\in \{2,3\}$, we are studying in this paper: the Simple Cubic (SC), triangular, Face-Centred-Cubic (FCC) and Body-Centred-Cubic lattices. The abbreviate names are given as well as their usual notations and their sets of minimal vectors expressed in the canonical basis $\{e_i\}_{1\leq i \leq d}$ of $\R^d$. Notice that $\Z^2$ will naturally be  called the square lattice (of unit length).}
\label{fig:lattices}
\end{table}

\begin{remark}
We recall that, since $\lambda \mathsf{A}_2$ and $\lambda \mathsf{D}_3$ are the lattices having 6 (resp. 12) bonds of length $\lambda$ in dimension $d=2$ (resp. $d=3$), the results proved in this paper are not really relevant for them. However, each of them belongs to the corresponding set of lattices $\mathcal{L}_d(\lambda)$, $d\in \{2,3\}$ with bond length $\lambda$. Notice that we also have $\lambda \mathsf{D}_3^*\in \overline{\mathcal{L}_3(\mathcal{M}_{\Z^3},\lambda)}$ and $\lambda \mathsf{D}_3\in \overline{\mathcal{L}_3(\mathcal{M}_{\mathsf{D}_3^*},\lambda)}$.
\end{remark}

To finish this section, we give the definition (see also \cite{FieldsEpstein}) of a type of lattices that plays a crucial role in our theoretical results.

\begin{defi}[Strongly eutactic lattices]\label{def:eutactfragile}
We say that a lattice $L\in \mathcal{L}_d$ is strongly eutactic if, for all $n\geq 1$ and for some positive scalars $\{\rho_n\}_{n\geq 1}$, 
$$
A_L^{-1}=\rho_n \sum_{|p|=\lambda_n(L)\atop p\in L} p^t p,\quad \textnormal{where}\quad Q_L(m)=m^t A_L m, \forall m\in \Z^d.
$$
\end{defi}
\begin{remark}[Spherical $2$-designs]
This definition is equivalent with the fact that all layers $D$ of $L$, i.e. the sets of points in the lattice belonging to the sphere $S_R$ centred at the origin and of radius $R$, are $2$-designs, up to a suitable renormalization (see e.g. \cite{Coulangeon:2010uq,Gruber}). This means that for any polynomial $P:\R^d\to \R$ of degree up to $2$, we have
$$
\frac{1}{|S_R|}\int_{S_R} P(x_1,...,x_d)dx=\frac{1}{\sharp D}\sum_{x\in D} P(x_1,...,x_d).
$$
\end{remark}

\subsection{Potentials and lattice energies}

We now define our spaces of radially symmetric interaction potentials and their associated energies.

\begin{defi}[Potentials and energy]\label{def:potenergy}
Let $d\geq 1$. We say that $f\in \mathcal{F}_d$ if $|f(r)|=O(r^{-\frac{d}{2}-\varepsilon})$ as $r\to +\infty$ for some $\varepsilon>0$ and if $f$ can be represented as the Laplace transform of a Radon measure $\mu_f$ on $(0,+\infty)$, i.e.
$$
f(r)=\int_0^{+\infty} e^{-rt}d\mu_f(t).
$$
Furthermore, we say that $f\in \mathcal{CM}_d$ if $f\in \mathcal{F}_d$ and $\mu_f$ is nonnegative.

\medskip

For any $f\in \mathcal{F}_d$, the $f$-energy of a lattice $L\in \mathcal{L}_d$ is defined by the absolutely convergent sum
\begin{equation}\label{eq:Ef}
E_f[L]:=\sum_{p\in L\backslash \{0\}} f(|p|^2).
\end{equation}
\end{defi}
\begin{remark}[Completely monotone potentials]
We notice that, by Hausdorff-Bernstein-Widder Theorem \cite{Bernstein-1929}, any $f\in \mathcal{CM}_d$ is completely monotone (which explains the notation), in the sense that, for all $r>0$ and all $k\in \N_0$, $(-1)^k f^{(k)}(r)\geq 0$. Also, as explained by Schoenberg in \cite{Schoenberg}, we know,  using Bochner's Theorem \cite{Bochner}, that $f\in \mathcal{CM}_d$  if and only if $r\mapsto f(r^2)$ is the Fourier transform of a nonnegative finite Borel measure on $\R$. 
\end{remark}
Two important lattice functions can be defined by \eqref{eq:Ef} for specific interaction potentials $f(r)\in \{ e^{-\pi \alpha r}, r^{-s/2} \}$: the lattice theta function $\theta_L(\alpha)$ and the Epstein zeta function $\zeta_L(s)$ respectively defined for any $\alpha>0$ and $s>d$, by
\begin{equation}\label{eq:zetatheta}
\theta_L(\alpha):=\sum_{p\in L} e^{-\pi \alpha |p|^2},\quad \textnormal{and}\quad \zeta_L(s):=\sum_{p\in L\backslash \{0\}} \frac{1}{|p|^s}.
\end{equation}
Notice that it is traditional to add the $p=0$ term to $\theta_L(\alpha)$ which makes no difference while minimizing $L\mapsto \theta_L(\alpha)$ among any type of lattices including the origin.

\medskip

The next definition concerns the notion of universal optimality originally introduced by Cohn and Kumar in \cite{CohnKumar} and showed in many contexts \cite{Mont,CKMRV2Theta,BFMaxTheta20,BDefects20}. Contrary to \cite{CohnKumar}, we do not stay in the set of configurations with fixed density but we simply say that an optimality property is universal if it holds for $E_f$ for any completely monotone function $f$ among lattices in a certain subset. In the following, the term ``optimal" can be replaced by any type of optimality: local/global minimality/maximality.

\begin{defi}[Universal optimality]\label{def:univopt}
Let $\mathcal{S}\subset \mathcal{L}_d$. We say that a lattice $L_0$ is universally optimal in $\mathcal{S}$ if one of these equivalent properties holds:
\begin{itemize}
\item for all $\alpha>0$, $L_0$ is optimal for $L\mapsto \theta_L(\alpha)$ in  $\mathcal{S}$;
\item for any $f\in \mathcal{CM}_d$, $L_0$ is optimal for $E_f$ in $\mathcal{S}$.
\end{itemize}
\end{defi}
\begin{remark}[Lattice theta function and universal optimality]
The equivalence stated in Definition \ref{def:univopt} is a straightforward application (see e.g. \cite{CohnKumar,BetTheta15}) of the fact that, by absolute summability,
$$
E_f[L]=\int_0^\infty \left[ \theta_L\left( \frac{\alpha}{\pi} \right)-1\right] d\mu_f(\alpha).
$$
Indeed, if $L_0$ is optimal for $L\mapsto \theta_L(\alpha)$ in $\mathcal{S}\subset \mathcal{L}_d$, $\forall \alpha>0$, then $L_0$ is optimal for $E_f$ where $f\in \mathcal{CM}_d$ since $\mu_f$ is nonnegative. The reciprocal implication is clear since $r\mapsto e^{-\pi \alpha r}\in \mathcal{CM}_d$ for all $\alpha>0$.
\end{remark}

\section{Universal minimality and consequences}\label{subsec:localstable}

Our first result is a generalization of Fields' work \cite{FieldsEpstein} who studied the Epstein zeta functions.

\begin{thm}[\textbf{Universal minimality of strongly eutactic lattices}]\label{thm:univlocal}
If $L_0$ is strongly eutactic then, for all $\alpha>0$, it is the only minimizer of $L\mapsto \theta_L(\alpha)$ in $\mathcal{L}_d(\mathcal{M}_{L_0},\lambda)$. In particular, all the strongly eutactic lattices $L_0$ are universally minimal in $\mathcal{L}_d(\mathcal{M}_{L_0},\lambda)$. Furthermore, in dimension $d\in \{2,3\}$, the only universally minimal lattices in $\mathcal{L}_d(\mathcal{M}_{L_0},\lambda)$ are $L_0\in\{ \lambda \Z^2,  \lambda \mathsf{A}_2, \lambda \Z^3, \lambda \mathsf{D}_3,  \lambda \mathsf{D}_3^*\}$.
\end{thm}
\begin{proof}
We follow the lines of \cite[Thm. 1]{FieldsEpstein}. We seek to minimize $L\mapsto \theta_L(\alpha)$ subject to the linear constraint, with respect to the coefficient of the quadratic form $Q_L(m)=\sum_{i,j} a_{i,j} m_i m_j$, written
\begin{equation}\label{eq:constraint}
Q_L(m_i)=\lambda^2\quad \forall i\in \{1,...,r\}, \quad \forall m_i\in \mathcal{M}:=\mathcal{M}_L.
\end{equation}
Let us show that any strongly eutactic lattice is a critical point of  $L\mapsto \theta_L(\alpha)$ on $\mathcal{L}_d(\mathcal{M},\lambda)$ for all $\alpha>0$. Being a critical point of $L\mapsto \theta_L(\alpha)$ for all $\alpha>0$ and satisfying the constraint \eqref{eq:constraint} means that there exists $\{\mu_k\}_k\subset \R^r$ such that, for all $i,j$ and all $\alpha>0$,
\begin{equation}\label{eq:criticpoint}
\partial_{a_{ij}}\left\{ \sum_{p\in \Z^d} e^{-\pi \alpha Q_L(p)})+\sum_{k=1}^r \mu_k Q_L(m_k) \right\}=0.
\end{equation}
We therefore obtain, for all $i,j$ and all $\alpha>0$,
$$
\sum_{p\in \Z^d} p_i p_j e^{-\pi \alpha Q_L(p)}=\sum_{m\in \mathcal{M}} \gamma_m m_i m_j,\quad \gamma_m:=\frac{\mu_{m_k}}{\pi \alpha} \textnormal{ where }  m=m_k.
$$
We hence get, for all $i,j,k,\ell$ and all $\alpha>0$,
$$
 \frac{\displaystyle\sum_{m\in \mathcal{M}} \gamma_m m_i m_j}{\displaystyle\sum_{m\in \mathcal{M}} \gamma_m m_k m_\ell}=\frac{\displaystyle\sum_{p\in \Z^d} p_i p_j e^{-\pi \alpha Q_L(p)}}{\displaystyle\sum_{p\in \Z^d} p_k p_\ell e^{-\pi \alpha Q_L(p)}}=\frac{\displaystyle \sum_{n=1}^\infty \sum_{|p|=\lambda_n(L)} p_i p_j e^{-\pi \alpha |p|^2}}{\displaystyle \sum_{n=1}^\infty \sum_{|p|=\lambda_n(L)} p_k p_\ell e^{-\pi \alpha |p|^2}}.
$$
As we successively multiply numerator and denominator by $e^{-\pi \alpha \lambda_n(L)^2}$ and letting $\alpha$ going to infinity, we find that a necessary condition for \eqref{eq:criticpoint} to hold for all $\alpha>0$ is that, for all $i,j,k,\ell$ and all $n\geq 1$,
$$
\frac{\displaystyle \sum_{|p|=\lambda_n(L)} p_i p_j}{\displaystyle \sum_{|p|=\lambda_n(L)} p_k p_\ell}= \frac{\displaystyle\sum_{m\in \mathcal{M}} \gamma_m m_i m_j}{\displaystyle\sum_{m\in \mathcal{M}} \gamma_m m_k m_\ell}.
$$
Writing $V_n:=\sum_{|p|=\lambda_n(L)}p^{t} p$, it means that $V_1=\rho_n V_n$ for all $n\geq 2$ for some scalars $\{\rho_n\}_{n\geq 2}$.  This is by definition the case if $L$ is strongly eutactic (see Definition \ref{def:eutactfragile}). It is simple to check (see also \cite{Gruber}) that it implies the same for any fixed $\alpha>0$. It follows that any strongly eutactic lattice is a critical point of  $L\mapsto \theta_L(\alpha)$ on $\mathcal{L}_d(\mathcal{M},\lambda)$ for all $\alpha>0$.

\medskip

Let us show that $L\mapsto \theta_L(\alpha)$ is convex in $\mathcal{L}_d(\mathcal{M},\lambda)$ for all $\alpha>0$. Indeed, let $ L\in \mathcal{L}_d(\mathcal{M},\lambda)$. We now show that, if $H=(h_{ij})\neq 0$ and writing $Q_{L}[m]=m^t A_{L} m=\sum_{i,j} a_{ij} m_i m_j$, then, for all $\alpha>0$
$$
 m^t (A_{L}+ H) m=\lambda^2, \forall m\in \mathcal{M} \Rightarrow \sum_{i,j,k,\ell} \frac{\partial^2}{\partial a_{ij}\partial a_{k \ell}}\theta_{L}(\alpha) h_{ij}h_{kl}>0.
$$
Again, the same arguments as in \cite[Thm. 1]{FieldsEpstein} works since our constraint is linear, remarking that, for all $\alpha>0$,
\begin{align*}
\sum_{i,j,k,\ell}\frac{\partial^2}{\partial a_{ij}\partial a_{k \ell}}\theta_{L}(\alpha) h_{ij}h_{kl}&=\alpha^2 \pi^2\sum_{i,j,k,\ell} \sum_{p\in L} p_i p_j p_k p_\ell e^{-\pi \alpha |p|^2}h_{ij} h_{k\ell}\\
&=\alpha^2\pi^2 \sum_{n=0}^\infty \sum_{|p|=\lambda_n(L)}\sum_{i,j,k,\ell}p_i p_j p_k p_\ell e^{-\pi \alpha |p|^2}h_{ij} h_{k\ell}\\
&=\alpha^2 \pi^2 \sum_{n=0}^\infty \sum_{|p|=\lambda_n( L)}\left(\sum_{i,j} p_i p_j h_{ij} e^{-\pi \alpha \frac{\lambda_n(L)^2}{2}} \right)\left(\sum_{k,\ell} p_k p_\ell h_{k \ell} e^{-\pi \alpha \frac{\lambda_n(L)^2}{2}} \right)\\
&=\alpha^2 \pi^2 \sum_{n=0}^\infty \sum_{|p|=\lambda_n(L)}\left(\sum_{i,j} p_i p_j h_{ij} e^{-\pi \alpha \frac{\lambda_n(L)^2}{2}} \right)^2>0,
\end{align*}
since $\sum_{i,j} p_i p_j h_{ij}=0$ for every $p\in L$ would imply that $H=0$. Therefore, the Hessian of the lattice theta function is positive definite for all lattices in the interior of a convex set (given by the linear constraints), which means that $L\mapsto \theta_L(\alpha)$ is convex and any critical point is a global minimizer. Notice that this proof follows from the linearity of the constraint combined with the additivity property of the exponential function.

\medskip

Since any universal minimizer is also locally minimal for any Epstein zeta function, we finally refer to \cite{FieldsEpstein,Gruber} to conclude that the only universally minimal lattices in dimensions $d\in \{2,3\}$ are the square, triangular, SC, BCC and FCC lattices, which completes the proof.
\end{proof}

Thus, Theorem \ref{thm:univlocal} gives a new global optimality result of all the SC lattices $\lambda \Z^d$ in $\mathcal{L}_d(\{\pm e_i\}_{i},\lambda)$ -- that are known to be strongly eutactic (see e.g. \cite[Sec. 4]{CoulLazzarini}) -- as well as the first minimality result for the BCC lattice $\lambda \mathsf{D}_3^*$ in $\mathcal{L}_3(\{\{\pm e_i\}_{1\leq i\leq 3}, \pm (e_1+e_2+e_3) \},\lambda)$. 

\begin{remark}[Properties of universal minimizers - Parametrization]\label{rmk:propunivloc}
Since any universally minimal lattice in $\mathcal{L}_d(\mathcal{M},\lambda)$ is also locally minimal for $L\mapsto \zeta_L(s)$, it follows by  \cite[Cor. 2]{FieldsEpstein} that there is only a finite number of such lattices in $\R^d$ and, by \cite[Cor. 1]{FieldsEpstein} that any such lattice has $d$ linearly independent minimum vectors. It is for example not surprising that no lattice with exactly 2 bonds can be a minimizer, since it is sufficient to infinitely elongate the other vectors in order to decrease the energy. This also implies that any universal (local) minimizer belongs to $\overline{\mathcal{L}_d(\mathcal{M},\lambda)}$ which can be written as the set of lattices $\{ L_t\}_{t\in T}$ where the set of parameters $T$ (typically a set of angles) can be assumed to be compact. For instance, the set $\mathcal{L}_2(\lambda)=\overline{\mathcal{L}_2(\mathcal{M}_{\Z^2},\lambda})$ can be viewed -- as a consequence of the quadratic form reduction method (see e.g. \cite{Mont,Terras_1988}) -- as the family of 2d lattices $\{\lambda L_t : t\in T:=[\pi/3,\pi/2]\}$ where $L_t=\Z(1,0)\oplus \Z( \cos t, \sin t)$, i.e. $\lambda L_t$ is a general rhombic lattice (including the square and the triangular one) of side length $\lambda$ and angle $t$. 
\end{remark}

A simple consequence of Theorem \ref{thm:univlocal} gives the following result. 

\begin{corollary}[\textbf{Comparison of important lattices 
- Universal maximizers}]\label{cor:2d3d}
For all $\lambda>0$, all $L\in \mathcal{L}_2(\mathcal{M}_{\Z^2},1)$, all $L'\in \mathcal{L}_3(\mathcal{M}_{\Z^3},1)$, all $L''\in \mathcal{L}_3(\mathcal{M}_{\mathsf{D}_3^*},1)$,  all $f\in \mathcal{CM}_2$ and all $g\in \mathcal{CM}_3$,
$$
E_f[\lambda \Z^2]\leq E_f[\lambda L]<E_f[\lambda \mathsf{A}_2],\quad \textnormal{and}\quad E_g[\lambda \Z^2]\leq E_g[\lambda L']<E_g[\lambda \mathsf{D}_3^*]\leq E_g[\lambda L'']<E_g[\lambda \mathsf{D}_3].
$$
In particular, $\lambda \mathsf{A}_2$ and $\lambda \mathsf{D}_3$ are universal maximizers respectively in $\mathcal{L}_2(\lambda)$ and $\mathcal{L}_3(\lambda)$.
\end{corollary}
\begin{proof}
The first double inequality comes from the universal minimality of $\lambda \Z^2$ in $\mathcal{L}_2(\{\pm e_i \}_i, \lambda)$ as well as the fact that the triangular lattice, which is the only one that has $6$ bonds (the maximum possible), reaches the maximum of $E_f$ in $\mathcal{L}_2(\lambda)$ as explained in Remark \ref{rmk:kiss}. The second chain of inequalities is a combination of the universal minimality of the SC lattice $\lambda \Z^3$ in $\mathcal{L}_3(\{\pm e_1, \pm e_2, \pm e_3 \}, \lambda)$ (which gives the first double inequality) and the one of the BCC lattice $\lambda \mathsf{D}_3^*$ in $\mathcal{L}_3(\{\{\pm e_i\}_i, \pm (e_1+e_2+e_3) \},\lambda)$ (which gives the last double inequality). Moreover, the universal maximality of $\lambda \mathsf{D}_3$ is ensures by the fact that it is the only lattice that has 12 bonds (the maximum possible again).
\end{proof}

In particular, contrary to the fixed density case where the FCC (resp. BCC) lattice is expected to be a global minimizer of $L\mapsto \theta_L(\alpha)$ for all $\alpha\geq 1$ (resp. $\alpha\leq 1$) according to Sarnak-Strombergsson Conjecture \cite[Eq. (43)]{SarStromb}, there is no phase transition of this type in the prescribed bonds case.

\begin{remark}[Generalization to convex potentials]
 Notice that the optimality result for $\lambda \Z^2$ can be easily generalized to convex potentials $f$ such that $|f(r)|=O(r^{-1-\varepsilon})$ as $r\to +\infty$ for some $\varepsilon>0$. We also expect Theorem \ref{thm:univlocal} and then Corollary \ref{cor:2d3d} to hold for convex functions as well, in any dimension $d$. Since we do not have the proof of such general property, we have chosen to not include the two-dimensional proof to this work.
\end{remark}

\section{Application to the Lennard-Jones type potentials}\label{subsec:LJ}

\subsection{Theoretical optimality result}
We now show that the result proved in Theorem \ref{thm:univlocal} also holds at small scale, i.e. for small $\lambda$, when $f$ is a Lennard-Jones type potential defined by
\begin{equation}\label{eq:LJdef}
f(r)=\frac{a}{r^{p}}-\frac{b}{r^q},\quad (a,b)\in (0,\infty),\quad p>q>d/2.
\end{equation}
Contrary to the unit density case where the space of lattices is not compact and for which we derived in \cite{BetTheta15,BDefects20} methods based on Laplace transforms in order to find a range of densities for which a lattice is minimal, the prescribed bonds case is much more easier. Indeed, since our space of lattices can be viewed as compact (see Remark \ref{rmk:propunivloc}), one can get with any degree of accuracy the sharp bounds on $\lambda$ (called $\lambda_0,\lambda_1$ in the following proposition) ensuring the optimality in $\overline{\mathcal{L}(\mathcal{M},\lambda)}$ of any universal optimizer $L$ (see Remark \ref{rmk:LJlambda0}).

\begin{prop}[\textbf{Optimality for Lennard-Jones type potentials}]\label{prop:LJ}
Let $f$ be a Lennard-Jones type potential defined by \eqref{eq:LJdef}. Let $\lambda L_0$ be the unique universal minimizer in $\mathcal{L}_d(\mathcal{M}_{L_0},\lambda)$. Then there exists $\lambda_0$ such that $\lambda L_0$ is minimal in $\mathcal{L}_d(\mathcal{M}_{L_0},\lambda)$ if and only if $\lambda \in (0,\lambda_0)$.

\noindent Furthermore, if $L_1$ is a unique universal maximizer, possibly local, in $\overline{\mathcal{L}_d(\mathcal{M},\lambda)}$ for some $\mathcal{M}\subset (\Z^d)^r$ (or in $\mathcal{L}_d(\lambda)$), then there exists $\lambda_1$ such that $\lambda L_1$ is (possibly locally) minimal in $\overline{\mathcal{L}_d(\mathcal{M},\lambda)}$ (or in $\mathcal{L}_d(\lambda)$) if and only if $\lambda>\lambda_1$.
\end{prop}
\begin{proof}
We write $\overline{\mathcal{L}_d(\mathcal{M}_{L_0},1)}=\{ L_t\}_{t\in T}$ where $T$, by universal minimality, is a compact set of parameters (see Remark \ref{rmk:propunivloc}) and $L_{t_0}=L_0$. Let us prove the first part of the proposition. We have
$$
E_f[\lambda L_t]- E_f[\lambda L_{t_0}]\geq 0,\forall t\in T \iff \lambda \leq \lambda_0:=\left(\frac{a}{b}  \right)^{\frac{1}{2(p-q)}}\inf_{t\in T\atop t\neq t_0}\left( \frac{\zeta_{L_t}(2p)- \zeta_{L_{t_0}}(2p)}{\zeta_{L_t}(2q)- \zeta_{L_{t_0}}(2q)} \right)^{\frac{1}{2(p-q)}}.
$$
Let us show that $\lambda_0>0$. By Taylor expansion of $\zeta_{L_t}(2x)$, $x\in \{p,q\}$, when $L_t$ is close to $L_{t_0}$ combined with Theorem \ref{thm:univlocal} that ensures the positivity of both numerator and denominator, we obtain that 
$$
\lim_{t\to t_0 \atop t\neq t_0}\frac{\zeta_{L_t}(2p)- \zeta_{L_{t_0}}(2p)}{\zeta_{L_t}(2q)- \zeta_{L_{t_0}}(2q)}=c_{p,q}>0.
$$
Therefore, since $T$ is compact we obtain by continuity that $\inf_{t\in T\atop t\neq t_0}\left( \frac{\zeta_{L_t}(2p)- \zeta_{L_{t_0}}(2p)}{\zeta_{L_t}(2q)- \zeta_{L_{t_0}}(2q)} \right)^{\frac{1}{2(p-q)}}$ is actually achieved for a certain parameter $\tilde{t}\in T$ for which the quotient cannot be zero since the numerator and the denominator, that are nonnegative, only cancel for $t=t_0$.

\medskip

The proof of the second part of the proposition is shown the same way since, for $L_{t_1}:=L_1$,
$$
E_f[\lambda L_t]- E_f[\lambda L_{t_1}]\geq 0,\forall t\in T \iff \lambda \geq \lambda_1:=\left(\frac{a}{b}  \right)^{\frac{1}{2(p-q)}}\sup_{t\in T \atop t\neq t_1} \left(\frac{\zeta_{L_{t_1}}(2p)- \zeta_{L_t}(2p)}{\zeta_{L_{t_1}}(2q)- \zeta_{L_t}(2q)} \right)^{\frac{1}{2(p-q)}}>+\infty.
$$
Indeed, both numerator and denominator of the above expression are positive and vanishes only if $t=t_1$. It follows again by Taylor expansion near $t_1$ and continuity that $\lambda_1$ is finite, the supremum being replaced by a maximum which is reached.

\end{proof}
\begin{remark}[Example of computation of $\lambda_0$ an $\lambda_1$]\label{rmk:LJlambda0}
We give an example of application of Proposition \ref{prop:LJ} in dimension $d\in \{2,3\}$ where we know the universal minimizers and maximizers by Corollary \ref{cor:2d3d}. 

\medskip

In dimension $d=2$, we know that $\lambda \Z^2$ is universally minimal in $\mathcal{L}_2(\{\pm e_i \}_i,\lambda)$ and $\lambda \mathsf{A}_2$ is universally maximal in $\mathcal{L}_2(\lambda)$ (see Corollary \ref{cor:2d3d}). Choosing $T=[\pi/3,\pi/2]$ (see Remark \ref{rmk:propunivloc}), let us write $\Z^2=L_{\frac{\pi}{2}}$. Therefore, the proof of Proposition \ref{prop:LJ} shows that the following quantities are finite and strictly positive:
\begin{align*}
&\lambda_0(\Z^2,p,q,a,b):=\inf_{t\in [\pi/3,\pi/2)} \left( \frac{a}{b}\right)^{\frac{1}{2(p-q)}}\left( \frac{\zeta_{L_t}(2p)- \zeta_{L_{{\frac{\pi}{2}}}}(2p)}{\zeta_{L_t}(2q)- \zeta_{L_{\frac{\pi}{2}}}(2q)} \right)^{\frac{1}{2(p-q)}},\\
&\lambda_1(\mathsf{A}_2,p,q,a,b):=\sup_{t\in (\pi/3,\pi/2]} \left(\frac{a}{b}  \right)^{\frac{1}{2(p-q)}} \left(\frac{\zeta_{L_{\frac{\pi}{3}}}(2p)- \zeta_{L_t}(2p)}{\zeta_{L_{\frac{\pi}{3}}}(2q)- \zeta_{L_t}(2q)} \right)^{\frac{1}{2(p-q)}}.
\end{align*}

For instance, if $(p,q,a,b)=(6,3,1,2)$, then $f(r^2)=r^{-12}-2 r^{-6}$ is the classical Lennard-Jones potential and we numerically compute that $\lambda_0(\Z^2,6,3,1,2)\approx 0.7628683$, for the optimality in $\mathcal{L}_2(\mathcal{M}_{\Z^2},\lambda)$ and $\lambda_1(\mathsf{A}_2,6,3,1,2)\approx 0.989$  in $\mathcal{L}_2(\lambda)$  (see also Figure \ref{fig:LJ2d} for another numerical confirmation).
\medskip

In dimension $d=3$, we use the same kind of formulas. We know that $\lambda \Z^3$ is universally minimal in $\mathcal{L}_3(\{\pm e_i \}_i,\lambda)$ and $\lambda \mathsf{D}_3$ is universally maximal in $\mathcal{L}_3(\lambda)$ (see Corollary \ref{cor:2d3d}). We find $\lambda_0(\Z^3,6,3,1,2)\approx 0.735$ for the optimality of the SC lattice $\lambda \Z^3$ in $\mathcal{L}_3(\mathcal{M}_{\Z^3},\lambda)$, $\lambda_1(\mathsf{D}_3,6,3,1,2)\approx 0.95$ for the minimality of the FCC lattice $\lambda \mathsf{D}_3$ in $\overline{\mathcal{L}_3(\mathcal{M}_{\mathsf{D}_3^*},\lambda)}$  and $\lambda_0(\mathsf{D}_3^*,6,3,1,2)\approx 0.94$ for the optimality of the BCC lattice $\lambda \mathsf{D}_3^*$ in $\mathcal{L}_3(\mathcal{M}_{\mathsf{D}_3^*},\lambda)$. 
\end{remark}

\subsection{Numerical investigation in dimensions 2 and 3}\label{subsec:num}

We finish this paper by investigating numerically the minimum of
$$
E_f[L]:=\zeta_L(12)-2\zeta_L(6),\quad \textnormal{where}\quad f(r)=\frac{1}{r^6}-\frac{2}{r^3}.
$$
among lattices of $\mathcal{L}_d(\lambda)$, $d\in \{2,3\}$. In this case, $r\mapsto f(r^2)=r^{-12}-2 r^{-6}$ is the classical (12,6) Lennard-Jones potential with minimum at $r_0=1$. All our numerical values are found using the software Scilab by only considering few (main) terms of the infinite sums defining $E_f$ since the inverse power laws we are considering are converging rapidly to zero. The threshold values of $\lambda$ -- at which the minimizers change -- are easily computed using a gradient descent method.

\medskip

In dimension $d=2$, we consider the family of 2d lattices $\{\lambda L_t : t\in T:=[\pi/3,\pi/2]\}$ where 
$$
L_t=\Z(1,0)\oplus \Z( \cos t, \sin t),
$$
i.e. $\lambda L_t$ is a general rhombic lattice (including the square and the triangular one) of side length $\lambda$ and angle $t$. This family corresponds to our set $\mathcal{L}_2(\lambda)$ (see Remark \ref{rmk:propunivloc}). We therefore have
$$
E_f[\lambda L_t]=\sum_{m,n} f\left(\lambda^2[ m^2 + n^2 + 2 mn \cos t ] \right).
$$

Our observations are the following concerning the two-dimensional case (see also Figure \ref{fig:LJ2d}):
\begin{itemize}
\item \textit{Optimality of the square lattice:} there exists $\lambda_0\approx 0.76286$ such that if $\lambda \in (0,\lambda_0)$, then $t=\pi/2$ is the unique minimizer of $t\mapsto E_f[L_t]$ in $T$, i.e. the square lattice is the unique minimizer of the energy. This is simply a confirmation of our findings in Proposition \ref{prop:LJ} and Remark \ref{rmk:LJlambda0}.
\item \textit{Optimality of rhombic lattices:} there exists $\lambda_1\approx 0.99$ such that if $\lambda \in (\lambda_0,\lambda_1)$, the minimizer of $t\mapsto E_f[L_t]$ is in $T\backslash \{\pi/3,\pi/2\}$ and increases when $\lambda$ increases, i.e. the minimizer of the energy is rhombic with an angle changing from (a bit larger than) $\pi/3$ to (a bit smaller than) $\pi/2$.
\item \textit{Optimality of the triangular lattice:} for all $\lambda>\lambda_1$, the minimizer of $t\mapsto E_f[L_t]$ in $T$ is $t=\pi/3$, i.e. the triangular lattice minimizes the energy. This is again simply a confirmation of our findings in Proposition \ref{prop:LJ} and Remark \ref{rmk:LJlambda0}.
\item \textit{Global minimizer among all $t$ and $\lambda>0$.} The global minimizer among lattices with bonds of side length $\lambda$ is $\lambda_{opt}\mathsf{A}_2$ where $\lambda_{opt}\approx 1.112$. This also confirms our numerical results in the fixed density case found in \cite{BetTheta15,Beterloc}.
\end{itemize}

\medskip

\begin{figure}[!h]
\centering
 \includegraphics[width=5.4cm]{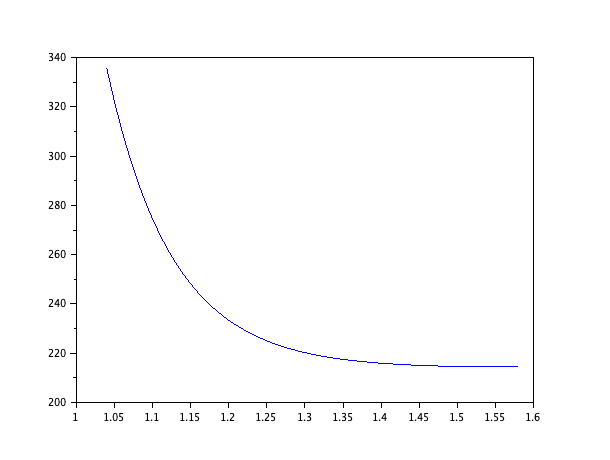} \quad  \includegraphics[width=5.4cm]{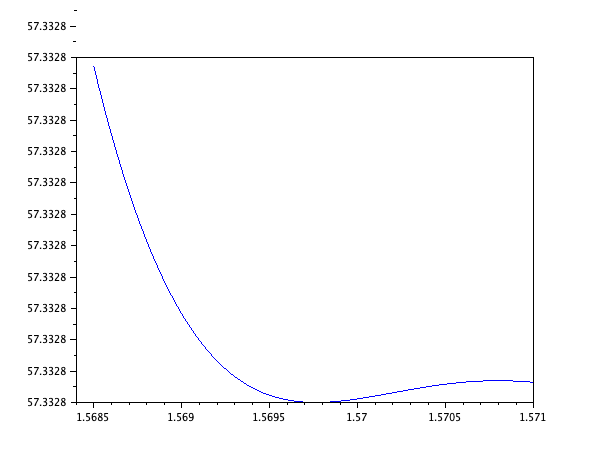}\quad 
 \includegraphics[width=5.4cm]{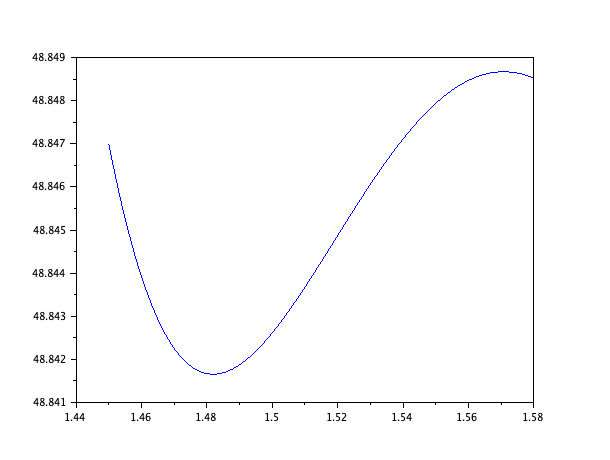}\quad
 \includegraphics[width=5.4cm]{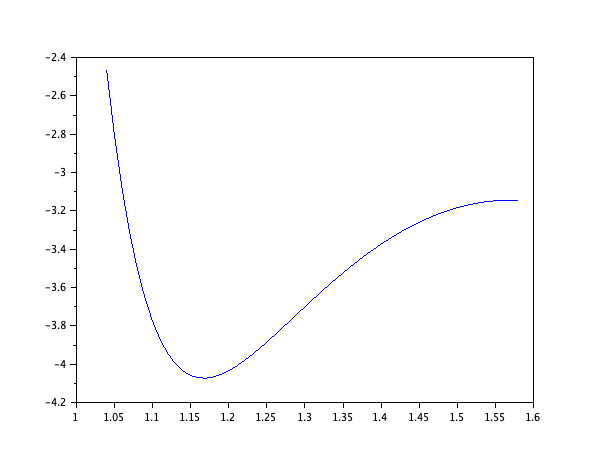}\quad \includegraphics[width=5.4cm]{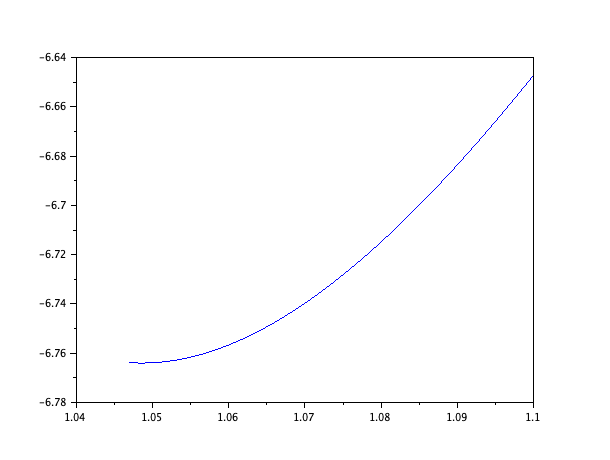} \quad \includegraphics[width=5.4cm]{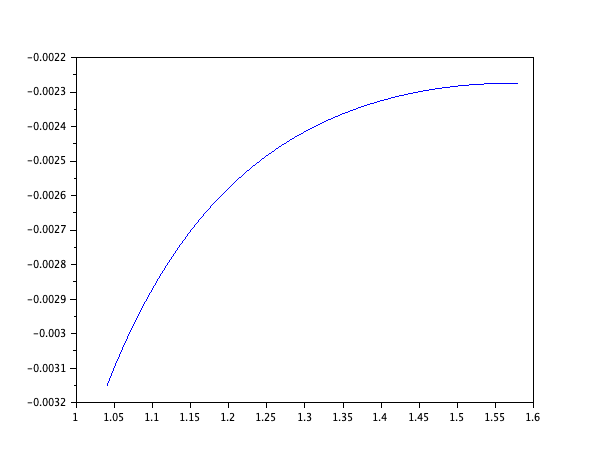} 
\caption{Plot of $t\mapsto E_f[\lambda L_t]$ on (sometimes a subset) of $[\pi/3,\pi/2]$ for different values of $\lambda$. For $\lambda\in (0,\lambda_0)$, e.g. $\lambda=0.7$ (up left), the minimizer is a square lattice (i.e. $t=\pi/2$). For $\lambda=0.7628684$ (up middle) the minimizer becomes a rhombic lattice (i.e. $t< \pi/2$) and stays rhombic for $\lambda\in (\lambda_0,\lambda_1)$, e.g. $\lambda=0.77$ (up right) and $\lambda=0.9$ (down left). Then for $\lambda>\lambda_1\approx 0.989$ (down middle) the minimizer becomes triangular (i.e. $t=\pi/3$) and stays like that for larger $\lambda$, e.g. for $\lambda=4$ (down right).}
\label{fig:LJ2d}
\end{figure}

\begin{remark}[Phase transitions for lattice ground states]
This phase transition depicted in Figure \ref{fig:LJ2d} is comparable to the one obtained in the two-dimensional Lennard-Jones case at fixed density where a transition of type ``triangular-rhombic-square-rectangular" is observed (see \cite{Beterloc,SamajTravenecLJ}). In our case, the transition as $\lambda$ increases is not very surprising since the square (resp. triangular) lattice is optimal for small (resp. large) $\lambda$ and all our building block functions are smooth and convex. In the fixed density case, the set of lattices is much larger and this is not clear why such minimizer must be in the boundary of the fundamental domain for 2d lattices (i.e. must be either rhombic or orthorhombic).

The same kind of behavior has been observed in \cite{LBMorse} to hold for Morse interaction potentials. Furthermore, Luo and Wei proved such kind of phase transition for sums of Coulombian interactions \cite{LuoChenWei} and sums of theta functions \cite{LuoWeiBEC} in the context of 3-block copolymers and two-components Bose-Einstein condensates (see also \cite{Mueller:2002aa} for a numerical investigation).
\end{remark}

\medskip

In dimension $d=3$, we parametrize our lattice $\lambda L_{t,\theta,\varphi}\in \mathcal{L}_3(\lambda)$ by three angles $(t,\theta,\varphi)$ in such a way that $L_{t,\theta,\varphi}=\Z u \oplus \Z v \oplus \Z w$ where
$$
u=(1,0,0),\quad v=(\cos t, \sin t, 0),\quad w=(\sin \theta \cos \varphi, \sin \theta \sin \varphi, \cos \theta).
$$
Therefore, we us the following expression to find our minimizers:
$$
E_f[\lambda L_{t,\theta,\varphi}]:=\sum_{m,n,k} f\left(\lambda^2\left[m^2 + n^2 + k^2 + 2 mn \cos t + 2 mk \sin\theta \cos \varphi + 2nk \sin \theta \cos(t-\varphi)  \right] \right).
$$
As an illustration, in Figure \ref{fig:comparcubic} we have plotted at the same time the energy of all cubic lattices $\{\lambda \Z^3, \lambda \mathsf{D}_3, \lambda \mathsf{D}_3^*\}$ with respect to $\lambda$. Furthermore, our numerical findings are the following:
\begin{itemize}
\item \textit{Optimality of the SC lattice:} as shown in Proposition \ref{prop:LJ} and Remark \ref{rmk:LJlambda0}, the minimizer of $E_f$ in $\mathcal{L}_3(\lambda)$ is $\lambda\Z^3$ if and only if $\lambda<\lambda_0(6,3,2,1)\approx 0.735$. 
\item \textit{Optimality of the BCC lattice:} there exists $\lambda_1\approx 0.9$ and $\lambda_2\approx 0.94$ such that $\lambda \mathsf{D}_3^*$ is the minimizer of $E_f$ in $\mathcal{L}_3(\lambda)$ if $\lambda \in (\lambda_1,\lambda_2)$.
\item \textit{Optimality of the FCC lattice.} there exists $\lambda_3\approx 0.95$ such that $\lambda \mathsf{D}_3$ is the minimizer of $E_f$ in $\mathcal{L}_3(\lambda)$ if $\lambda >\lambda_3$. In particular, the global minimizer of $E_f$ among all lattices with bonds of side length $\lambda$ is the FCC lattice $\lambda_{opt}\mathsf{D}_3$ where $\lambda_{opt}\approx 0.97$. This again confirms what we have found in Proposition \ref{prop:LJ} and Remark \ref{rmk:LJlambda0}.
\item \textit{Optimality of rhombic lattices.} for $\lambda \in (\lambda_0, \lambda_1)\cup (\lambda_2,\lambda_3)$, the minimizer of $E_f$ in $\mathcal{L}_3(\lambda)$ is a rhombic lattice.
\end{itemize}

\begin{figure}[!h]
\centering
 \includegraphics[width=9cm]{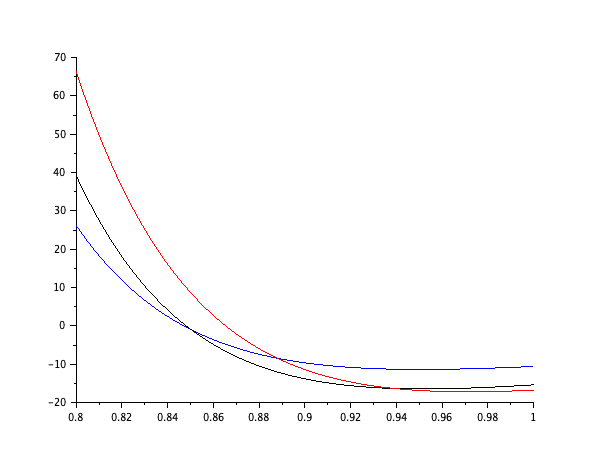}
\caption{Plot of $\lambda \mapsto E_f[\lambda L]$ on $[0.8,1]$ for $L=\Z^3$ (blue), $L=\mathsf{D}_3^*$ (black) and $L=\mathsf{D}_3$ (red).}
\label{fig:comparcubic}
\end{figure}

\noindent \textbf{Acknowledgement:} I would like to thank the Austrian Science Fund (FWF) for its financial support through the project F65 as well as  the DFG-FWF international joint project FR 4083/3-1/I 4354. I would like also to thank Ulisse Stefanelli for his suggestions concerning the first version of the draft. I am also grateful to the anonymous referees for their helpful comments.

\medskip

{\small \bibliographystyle{plain}
\bibliography{Biblio}}
\end{document}